\documentclass[conference]{IEEEtran}
\IEEEoverridecommandlockouts
\usepackage{cite}
\usepackage{amsmath,amssymb,amsfonts, amsthm}
\usepackage{algorithm}
\usepackage[table]{xcolor}
\usepackage{algpseudocode}
\usepackage{graphicx}
\usepackage{textcomp}
\usepackage{xcolor}
\usepackage{multirow}

\newtheorem{theorem}{Theorem}

\def\BibTeX{{\rm B\kern-.05em{\sc i\kern-.025em b}\kern-.08em
    T\kern-.1667em\lower.7ex\hbox{E}\kern-.125emX}}
\begin{document}

\title{Reed-Muller Codes for Joint Random and Stuck-At Error Correction\\
%{\footnotesize \textsuperscript{*}Note: Sub-titles are not captured for https://ieeexplore.ieee.org  and
}

%\author{\IEEEauthorblockN{1\textsuperscript{st} Ivana Djurdjevic}
%\and
%\IEEEauthorblockN{2\textsuperscript{nd} Robert Mateescu}
%\and
%\IEEEauthorblockN{3\textsuperscript{rd} Cyril Guyot}
\author{\IEEEauthorblockN{Ivana Djurdjevic, Robert Mateescu and Cyril Guyot}
\IEEEauthorblockA{\textit{WD Research}\\
San Jose, California, USA \\
Email: ivana.djurdjevic@wdc.com, robert.mateescu@wdc.com and cyril.guyot@wdc.com}
}

\maketitle

\begin{abstract}
Block codes are considered for improving the reliability of messages stored in a computer memory with both stuck-at defects and random errors. It is assumed that the side information about the state of the defects is available to the encoder, but not to the decoder. A novel recursive construction of a set of masks is developed such that it can satisfy any $s$ stuck-at errors in a $2^m$ binary sequence, when $s \leq m$. We prove that the masks generated in this way are codewords in a Reed-Muller $RM(s-1, m)$ code. The constructed set contains no more than $2^s m^{s-1}$ masks. We provide the lower and the upper bound on the size of the stuck-at redundancy, a fixed subset of mask bits that uniquely represents each mask in the set. The stuck-at code constructed in this way is a non-linear code. It is also a subcode of an $RM(r,m)$ code, with $ r \geq s-1$, that can be used for additional random error correction. The encoding requires no mask search and is straightforward based on the description of the recursive construction. The decoding is done in a single attempt and requires almost no additional complexity or latency.  
\end{abstract}

\begin{IEEEkeywords}
flash memory, non-volatile memory, stuck cells, linear error-correcting codes, BCH codes, Reed-Muller codes, side information, non-linear codes
\end{IEEEkeywords}

\section{Introduction}
In storage systems, media bit errors are an inevitable challenge that degrades the performance of the system. We distinguish two main types of errors. The first type are errors that arise from different types of random noise. These errors are intermittent and also referred to as soft errors. Another notable phenomenon is the occurrence of stuck-at errors or defects, where certain bits become fixed at a particular level (either 0 or 1) due to physical defects or malfunctioning hardware components. In order to maintain high reliability of the read messages, careful design of efficient and powerful error control coding techniques is essential. In the case of stuck bits, error correction techniques need to be more effective in order to account for the fact that some bits cannot be toggled.

Stuck-at errors typically arise in various non-volatile storage devices such as flash memory, phase-change memory (PCM), magnetoresistive memory (MRAM), and the like. These devices require specialized approaches to ensure data integrity when errors such as stuck bits occur. On the decoder side, techniques for decoding linear block codes with random errors and erasures can be directly applied together with various data recovery techniques. In some instances, the reading and decoding time is more critical than that of writing and encoding. In these situations, if additional knowledge about the positions of the stuck bits and the value to which they are stuck is available on the encoding side, it may be utilized to improve the performance of the coding system. This side information is used by incorporating stuck-at information in the encoding of linear block codes in a way that the decoder can decode the message correctly without any side information.

The coding for stuck-at errors originated in a Kuznetsov and Tsybakov work \cite{Kuznetsov-Tsybakov}. They consider coding for binary memories with stuck-at cells with the location and nature of the defects available to the encoder and not to the decoder. This work is expanded to coding for joint stuck-at and random errors in \cite{Kuznetsov_etal}, providing bounds on information rate and construction examples for $s=2$. Most of the prior work on joint random and stuck-at error correction focuses on partitioning a traditional linear code into disjoint subsets in such a way that for any combination of a certain number of stuck-at errors it is possible to find a codeword in one of these subsets that has target defect values at stuck-at  locations \cite{Heegard}. The extensive analysis for these types of approaches, both for binary and q-ary symbols can be found in \cite{kim2022codingboundspartiallydefective}. Other relevant stuck bit error correction work can be found in \cite{Borden-Vinck, Gabrys_etal, Wachter-Zeh_Yaakobi, Motwani_etal, con2024codefitsallstrong, dumer}.

Despite the significant advances in stuck bit coding, there are still several challenges that need to be addressed. Joint correction of binary stuck-at and random errors has been considered mostly using BCH codes \cite{Heegard} with extra redundancy for masking stuck-at defects being of the same order as the redundancy used to combat random errors, i.e. proportional to $s\log_2 n$. Furthermore, to the best of our knowledge, systematic construction of encoding masking set that can address arbitrary number of stuck-at errors $s$ together with an arbitrary number of random errors $t$ has not been developed. In this work, we aim to address some of these issues. In Section II, we define the problem and provide basic assumptions and terminology. In Section III, we present a recursive scheme that constructs a set of masks correcting any $1 \leq s \leq m$ stuck bit errors in a sequence of $2^m$ bits. It is proven that the number of masks does not exceed $2^s m^{s-1}$. This construction can be viewed as a generalization of the result in \cite{Kuznetsov-Tsybakov} for $s=2$ to any $s$. We also show that the masks belong to a Reed-Muller $RM(s-1, m)$ code. In Section IV, we describe the process of labeling masks in the set, i.e. determining the redundant bit positions. For $s=2$, we give an algorithm for finding all possible labels. For $s>2$, we use greedy search for finding a label and provide bounds on the number of label bits. The bounds show that the redundancy grows as $\log \log n = \log m$ for fixed $s$. In Section V, we describe how constructed set of masks can be used together with Reed-Muller code for joint stuck-at and random error correction. Section VI concludes the work and provides future directions.

\section{Problem Statement}

Consider an encoder for a joint stuck-at and random error correction, using an additive scheme for masking multiple defects as shown in Fig.~\ref{fig-block-diagram-jnt-stuck-rand-error-encoder}. The symbols written in the memory are binary. The redundant bits that eventually represent a stuck-at mask label are inserted in a binary message $\mathbf{u}$ at appropriate predefined locations to form an intermediate message $\mathbf{v}$. They are inserted as zeros initially as the mask is not determined yet. The intermediate message $\mathbf{v}$ is encoded using a systematic random error correction code to form an intermediate codeword $\mathbf{b}$. The mask generator block is provided the side information on the location and value of stuck-at defects in addition to the intermediate codeword $\mathbf{b}$. Based on this information, the mask $\mathbf{m}$ is generated such that it belongs to the code space of the random error correcting code used and that the final codeword $\mathbf{c}=\mathbf{m}+\mathbf{b}$ satisfies all stuck-at values. This process is called {\it masking the defect}. We also say that a particular mask {\it covers} particular stuck-at defects. Since all stuck-at positions are satisfied in the written sequence, on the decoder side shown in Fig.~\ref{fig-block-diagram-jnt-stuck-rand-error-decoder}, the masked message is read with no stuck-at errors and only contains random errors. The random error correcting code is used to correct random errors and determine the codeword estimate $\hat{\mathbf{c}}$. The label bits are used to lookup the estimate of the additive mask $\hat{\mathbf{m}}$. The decoder obtains an intermediate codeword estimate $\hat{\mathbf{b}}$ by adding the mask estimate to the codeword estimate. Finally, a message estimate $\hat{\mathbf{u}}$ is retrieved by removing the inserted mask label bits and random error correction code redundant parity bits. 

\begin{figure}[htbp]
\centerline{\includegraphics[width=8cm]{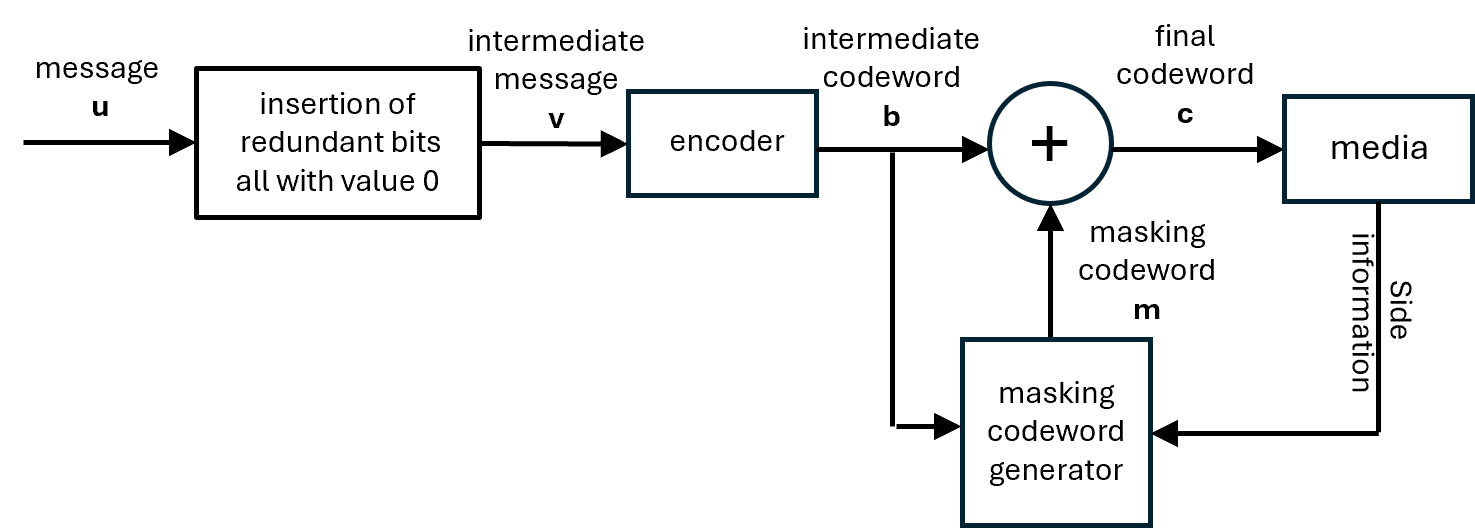}}
\caption{Block diagram of an encoder for joint random error correction and additive masking of stuck-at defects.}
\label{fig-block-diagram-jnt-stuck-rand-error-encoder}
\end{figure}

\begin{figure}[htbp]
\centerline{\includegraphics[width=8cm]{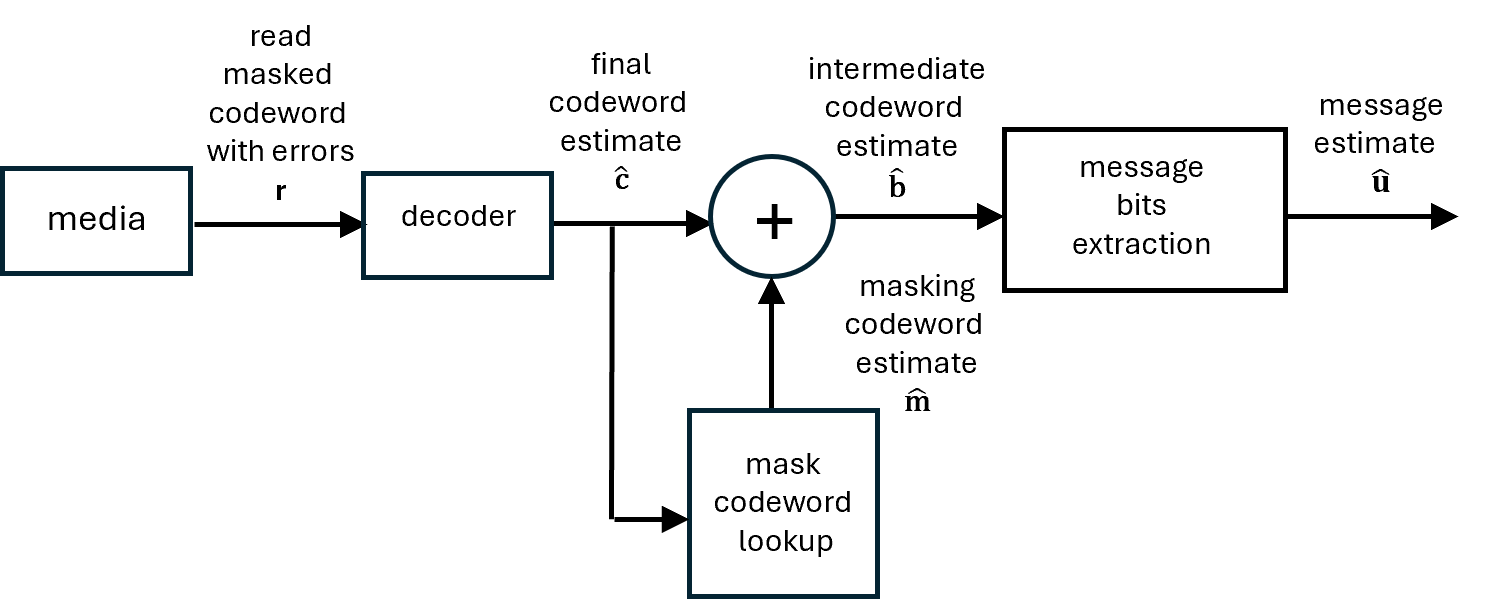}}
\caption{Block diagram of a decoder for joint random error correction and additive masking of stuck-at defects.}
\label{fig-block-diagram-jnt-stuck-rand-error-decoder}
\end{figure}

The details of encoding and decoding algorithms are given in Algorithm \ref{alg:encII} and Algorithm \ref{alg:decII}.

\begin{algorithm} 
\caption{Encoding}\label{alg:encII}
INPUT: 
\begin{algorithmic}[1]  
\State message $\mathbf{u} = \left[ u_0, u_1, \ldots, u_{k-L-1} \right]$ with $L$ being a number of label bits 
\State stuck-at locations $\{i_0, i_1, \ldots,  i_{s-1}\} \subset \{0,1,\ldots,n-1\}$ and corresponding stuck-at values $d_i \in \{0,1\}$ for all $i \in \{0,1,\ldots, s-1\}$ 
\State set $M$ of mask codewords in an $(n,k)$ systematic code $\mathcal{C}$ and label bit locations $\{l_0, l_1, \ldots, l_{L-1}\} \subset \{0,1,\ldots,k-1\}$ 
\end{algorithmic}
OUTPUT: masked codeword $\mathbf{c}$

\begin{algorithmic}[1]  
\State Form intermediate message $\mathbf{v} = \left[ v_0 \; v_1 \; \ldots \; v_{k-1} \right]$ by inserting zeros at label bit locations
\State Encode intermediate message using systematic $(n,k)$ code $\mathcal{C}$ to add $n-k$ parity bits and obtain an intermediate codeword $\mathbf{b}$
\State Find a mask $\mathbf{m}=\left[m_0 \; m_1 \; \ldots m_{n-1} \right]$ such that:
\begin{equation*}
b_i + m_i = d_i \;\; \forall i \in \{ i_0, i_1, \ldots, i_{s-1} \}
\end{equation*}
\State Compute final masked codeword: $\mathbf{c} = \mathbf{v} + \mathbf{m}$
\end{algorithmic}  
\end{algorithm}

\begin{algorithm} 
\caption{Decoding}\label{alg:decII}
INPUT: \\
read masked codeword corrupted with random errors $\mathbf{r} = \mathbf{c} + \mathbf{e}$ \\
OUTPUT: \\
message estimate $\hat{\mathbf{m}}$

\begin{algorithmic}[1]  
\State Decode read word $\mathbf{r}$ to correct random errors and obtain written codeword estimate $\hat{\mathbf{c}}$
\State Extract mask label: $\left(\hat{c}_{l_0}, \hat{c}_{l_1}, \ldots, \hat{c}_{l_{L-1}}\right)$
\State Find the mask $\hat{\mathbf{m}}$ that corresponds to the extracted label
\State Compute intermediate codeword estimate:
\begin{equation*}
\hat{\mathbf{b}} = \hat{\mathbf{c}} - \hat{\mathbf{m}} 
\end{equation*}
\State Remove label bits $\hat{b}_{l_0}, \hat{b}_{l_1}, \ldots, \hat{b}_{l_{L-1}}$ and $n-k$ parity bits from the intermediate codeword to obtain message estimate $\hat{\mathbf{u}}$.
\end{algorithmic}  
\end{algorithm}

In order to fully define the encoding and decoding process in Fig.~\ref{fig-block-diagram-jnt-stuck-rand-error-encoder} and Fig.~\ref{fig-block-diagram-jnt-stuck-rand-error-decoder}, we need to define the set of mask codewords or an algorithm for mask codeword generation such that for any $s$ stuck-at positions, any combination of stuck-at values in those $s$ positions and any input binary codeword $\mathbf{c}$ that belongs to a binary random error correcting $(n,k)$ code $\mathcal{C}$, an appropriate codeword mask $\mathbf{m} \in \mathcal{C}$ can be generated or found such that $\mathbf{c} + \mathbf{m}$ at $s$ stuck-at positions has correct stuck-at values. In the process of defining the mask set, it is important to minimize the required redundancy for a given $n$, $k$ and $s$. In particular, we would like to achieve redundancy that is better than the redundancy of a typical random error correcting code. For example, the redundancy of high rate BCH codes correcting $t$ errors is proportional to the $t \log_2 n$. The cardinality of the mask set determines the lower bound on the stuck-at redundancy. In addition to defining the mask set, selecting the positions used for mask labeling within the mask set also presents a challenge. The label size equals the redundancy needed for stuck-at defect masking and minimizing the label size is the ultimate goal that maximizes the overall code rate.

\subsection{Preliminaries: stuck-at defects only}

Consider the case of no random errors. Note that this is a special case of the scheme desribed in Fig. \ref{fig-block-diagram-jnt-stuck-rand-error-encoder} and Fig. \ref{fig-block-diagram-jnt-stuck-rand-error-decoder} when the masks in the set do not need to satisfy any random error correcting code constraints. First, assume a single stuck-at error. The set of masks is trivial and consists of an all-1 word and an all-0 word. The required redundancy is one extra bit regardless of the message length. Finding the smallest set of masks that covers any possible $2$ stuck-at defects is a harder problem. One such set is shown for the example of $n=8$ in Fig.~\ref{fig-masking-words-0}. The two initial masks, all-zero and all-one mask, cover any $(i,j)$ patterns of $(0,0)$ or $(1,1)$. The patterns $(0,1)$ and  $(1,0)$ are checked recursively for $(i,j)$ positions in different halves. In general, if $n = 2^m$, the total number of masks is $2(1+ m)$ and the required redundancy is at least $ \lceil 1+  \log_2 (m+1) \rceil$. This construction can be found in \cite{Kuznetsov-Tsybakov}. 

\begin{figure}[htbp]
\centerline{\includegraphics[width=3.5cm]{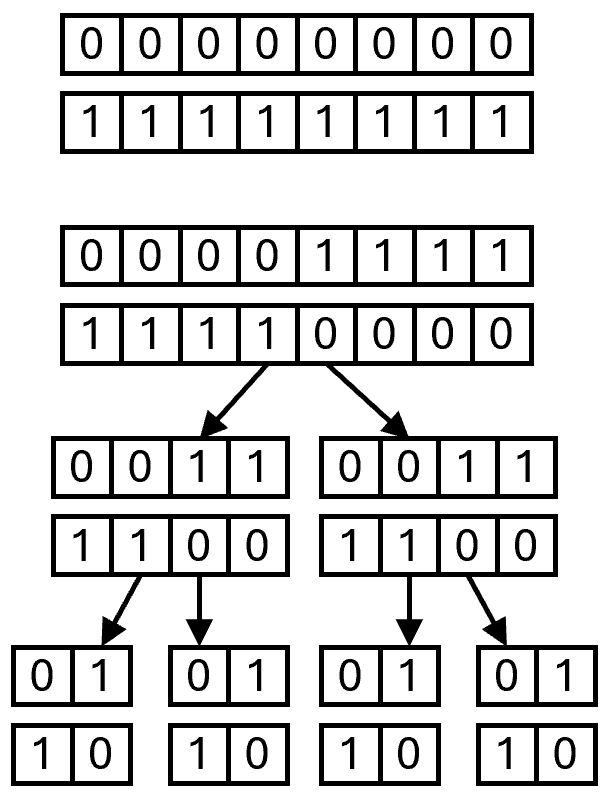}}
\caption{A set of 8 masks that can correct any 2 stuck-at defects in a message of length 8}
\label{fig-masking-words-0}
\end{figure}

For $s>2$ stuck-at defects, the problem of mask set construction, and especially a mask set that is a subset of a random error correcting code, is significantly harder. Bounds on the minimal redundancy $r_{\min}(n, s)$ of a code of length $n$ correcting stuck defects of multiplicity $s$, for $s \geq 2$, are as follows \cite{dumer}: 
\begin{IEEEeqnarray}{rCl}
\label{eq:KTbound}
\log \log n \leq r_{\min}(n, s) \leq  s + \lceil \log \ln  2^s \binom{n}{s}  \rceil
\end{IEEEeqnarray}
Therefore, redundancy $r(s,n) \sim \log \log n$, as $n \rightarrow \infty$ for $s = \mathrm{const}$, is considered asymptotically optimal for $s\geq 2$. The construction of asymptotically optimal encoding for stuck-at correction only is described in \cite{dumer}. The construction of joint random and stuck-at correction is given in \cite{Heegard}, but $\alpha(s,n) \sim \log n$ as $n \rightarrow \infty$ for $s = \mathrm{const}$. The code for random error correction is BCH and generated masks are in the BCH code space. In this work, we describe a novel construction of the mask set and we show that the mask set is in the Reed-Muller code space. We also show that the redundancy grows $\sim \beta(s) \log \log n$, as $n \rightarrow \infty$ for $s = \mathrm{const}$, with $\beta(s)$ being a constant depending on $s$.

\section{Recursive construction and properties of the mask set}

Let $M(s,m)$ denote the set of masks that can be used to cover any $s$ stuck-at defects in a binary sequence of length $2^m$. 
Consider the following recursive construction for any $s\geq2$:
\begin{IEEEeqnarray}{rCl}
\label{eq:recursive_union}
 M(s,m) &=&\{ [\mathbf{m}, \mathbf{m}] \; \forall \mathbf{m} \in M(s,m-1) \}   \nonumber \\
&\cup_{i=1}^{s-1}&  \{  [\mathbf{m}_1, \mathbf{m}_2]  \forall   \mathbf{m}_1 \in M(i,m-1),  \nonumber \\
&& \; \forall \mathbf{m}_2 \in M(s-i,m-1) \}. 
\end{IEEEeqnarray}
The final levels of recursion are: 
\begin{itemize}
\item $M(1,m-j)$, consisting of $2$ masks, one all-zero and one all-one mask, of length $2^{m-j}$, 
\item $M(s,l)$, with $l=\lceil log_2 s \rceil$, consisting of $2^{2^l}$ masks that are all binary $2^l$-tuples.
\end{itemize}
The first set in the union in (\ref{eq:recursive_union}) contains all masks that can cover any $s$ stuck bits in either left or right half of the binary message. The second set in the union in (\ref{eq:recursive_union}) contains all masks that can cover $i$ stuck bits in the left half and $s-i$ stuck bits in the right half of the binary message, for $i=1,2,\ldots, s-1$. 

An example of a set of masks correcting any $s=3$ stuck-at defects obtained from this recursive construction for binary sequence of length $2^3=8$ is given in Fig.~\ref{fig-example-recursive-s3m3-v2}. The block labeled with $0$ is an all-zero matrix and the block labeled with $1$ is an all-one matrix, both of appropriate dimensions matching neighboring blocks. The total number of masks of length $2^3=8$ in the example in Fig.~\ref{fig-example-recursive-s3m3-v2} is $40$. However, the number of distinct masks in the set $M(3,3)$ is $34$ and the redundancy required to represent this set of masks is at least $6$ bits.

Note that the set of masks covering any $2$ stuck-at errors in a sequence of length $2^m$ constructed in this way corresponds exactly to the construction given in \cite{Kuznetsov-Tsybakov}. 

\begin{figure}[htbp]
\centerline{\includegraphics[width=8.5cm]{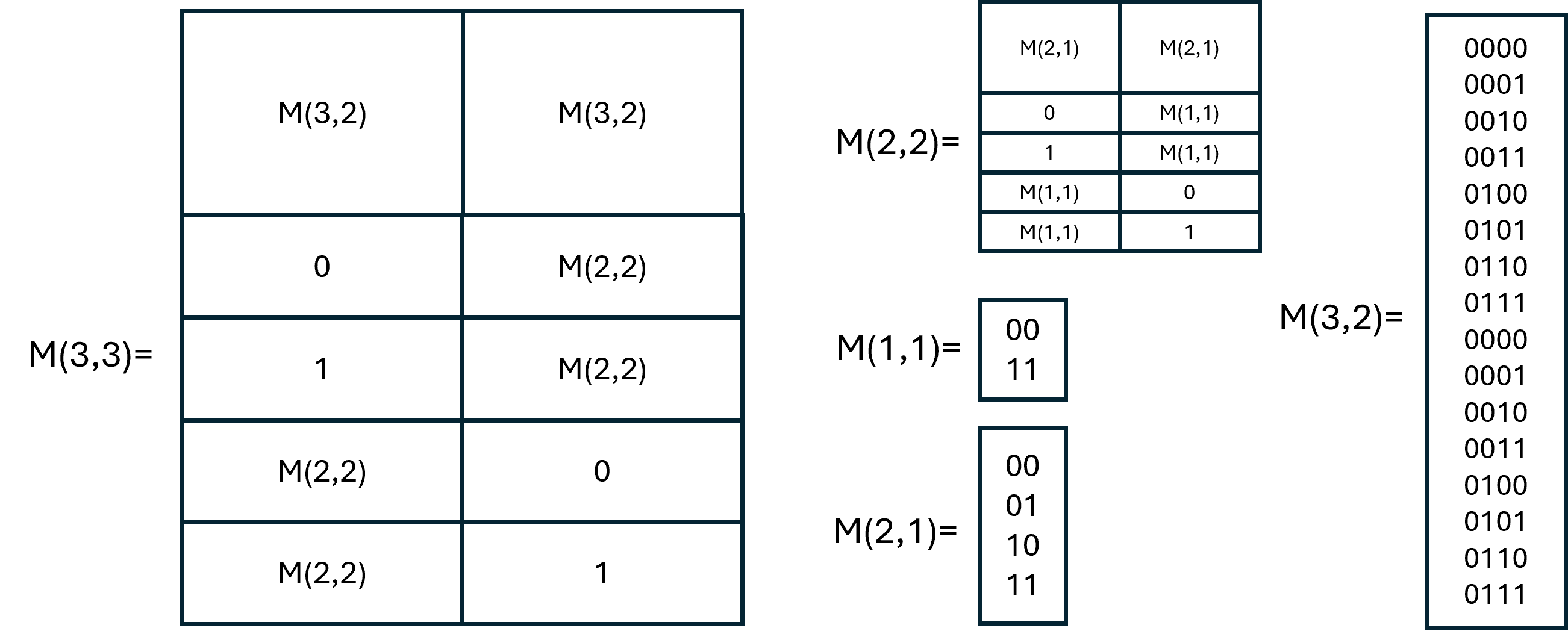}}
\caption{Recursive construction of masks of length 8 that can cover any $s=3$ stuck errors}
\label{fig-example-recursive-s3m3-v2}
\end{figure}

An upper bound on the number of masks in the set $M(s,m)$, denoted by $N_M(s,m)$, is given in the next theorem.
See Appendix~\ref{app:proof-thm-max-num-masks} for the full proof.

\begin{theorem}
\label{thm-max-num-masks}
The number of masks obtained using recursive construction in (\ref{eq:recursive_union}) is at most $2^s m^{s-1}$.
\end{theorem}

The next theorem shows the relationship between mask sets.
See Appendix~\ref{app:proof-thm-masks-subsets} for the full proof.
\begin{theorem}
\label{thm-masks-subsets}
The set of masks covering $s-1$ stuck bits in a binary sequence of length $2^m$ is a subset of the set of masks covering $s$ stuck bits in a binary sequence of the same length, i.e. $M(s-1,m) \subset M(s,m)$.
\end{theorem}

\subsection{Relationship to Reed-Muller codes}

Let us denote binary field as $\mathbb{F}_2$. A set of all binary $m$-dimensional vectors can be written as $X = \mathbb{F}_2^m = \{x_0,x_1,\ldots, x_{2^m-1} \}$. We define in $2^m$-dimensional space $\mathbb{F}_2^{2^m}$ the indicator vectors $\mathbb{I}_A \in \mathbb{F}_2^{2^m}$ on subsets $A \subset X$ by:
\begin{equation*}
( \mathbb{I}_A)_i =\begin{cases}
    1 & \text{if } x_i \in A \\\\
   0 & \text{if } otherwise
\end{cases}
\end{equation*}
A Reed-Muller code is defined by two parameters, $r$ and $m$. The first parameter $r$  is an order of the code and determines the code rate, while the second parameter $m$ defines the length of the code, $2^m$.
A Reed-Muller code of order $0$, $R(0,m)$ has a generator matrix with a single row that is an incidence vector of the entire space $ \mathbb{F}_2^m$. It is a repetition code of length $2^m$. The generator matrix of a $RM(1,m)$ code consists of $m+1$ row vectors of length $2^m$ where the first row is an all-one vector. The remaining $m$ rows are indicator vectors of $(m-1)$-dimensional subspaces of $m$-dimensional space. For example, let $m=3$. Then $n=8$ and 
\[ X = \mathbb{F}_2^3 = \{ (0,0,0), (0,0,1), (0,1,0), \ldots (1,1,1)\}. \]
The generator matrix of $RM(1,3)$ is:
\begin{equation*}
G = 
\begin{bmatrix}
1 & 1 & 1 & 1 & 1 & 1 & 1 & 1 \\
1 & 0 & 1 & 0 & 1 & 0 & 1 & 0 \\
1 & 1 & 0 & 0 & 1 & 1 & 0 & 0 \\
1 & 1 & 1 & 1 & 0 & 0 & 0 & 0 
\end{bmatrix}
\end{equation*}
A Reed-Muller code $RM(r,m)$ of order $r>1$ contains in its generator matrix indicator vectors of all possible intersections of at most $r$ different $(m-1)$-dimensional subspaces. However, the $RM(r,m)$ code can also be constructed recursively as follows:
\[ RM(r,m) = RM(r,m-1) | RM(r-1, m-1) \]
where $|$ denotes bar product of two codes defined as: $C_1 | C_2 = \{ (c_1 | c_1+c_2): c_1 \in C_1, c_2 \in C_2 \}$.
This recursive construction is also called Plotkin's construction and it is equivalent to the recursive expression for the generator matrix of $RM(r,m)$ code:
\begin{equation*}
G(r,m) = 
\begin{bmatrix}
G(r,m-1) & G(r,m-1) \\
0 & G(r-1, m-1) 
\end{bmatrix}
\end{equation*}
Note that code $RM(0,m)$ is a subset of code $RM(1,m)$. Using Plotkin's recursive construction and induction we can prove that Reed-Muller codes of the same length and increasing orders have the following relationship: $RM(0,m) \subset RM(1,m) \subset RM(2,m) \subset \cdots \subset RM(m-1,m) \subset RM(m,m)$.
Observe that $RM(m,m)$ is a code rate $1$ code with codewords being all $2^{2^m}$ binary $2^m$-tuples. The recursive construction and nested structure of Reed-Muller codes of the same length imply:
\begin{itemize}
\item $(c_i, c_i) \in RM(r,m)$ for any $c_i \in RM(r,m-1)$, and
\item $(c_i, c_j) \in RM(r,m)$ for any two codewords $c_i, c_j \in RM(r-1,m-1) \subset RM(r,m-1)$. Note that this is not true for any two codewords in $RM(r,m-1)$. 
\end{itemize}

A Reed-Muller code $RM(0,m)$ consists of only all-one and all-zero vector. Hence, $M(1,m) \subseteq RM(0,m)$.

Also, it is easy to see that masks constructed for $s=2$ stuck bits in \cite{Kuznetsov-Tsybakov} and shown in an example for $n=2^3$ in Fig.~\ref{fig-masking-words-0} are codewords in Reed-Muller code $RM(1,m)$. In fact, they are rows of $R(1,m)$ code's generator matrix and their complements. Therefore, $M(2,m) 
\subset RM(1,m)$. 

A relationship between masks of the construction in (\ref{eq:recursive_union}) and Reed-Muller code and its generator matrix exists for any $s$, as stated in the following two theorems.
See Appendix~\ref{app:proof-thm-masks-in-RMcode} and Appendix~\ref{app:proof-thm-masks-in-G-RMcode} for the full proofs.
\begin{theorem}
\label{thm-masks-in-RMcode}
The set of masks covering $s$ stuck bits in a binary sequence of length $2^m$ generated by the recursion in (\ref{eq:recursive_union}) is a subset of the set of codewords of Reed-Muller code of order $s-1$ and length $2^m$, $RM(s-1,m)$.
\end{theorem}

\begin{theorem}
\label{thm-masks-in-G-RMcode}
The set of masks covering  $s$ stuck bits in a binary sequence of length $2^m$ generated by the recursion in (\ref{eq:recursive_union}) contains all the rows of the generator matrix of a Reed-Muller $RM(s-1,m)$ code and their complements. Hence, we have:
\begin{equation}
N_M(s,m) \geq 2 \sum_{i=0}^{\min(m,s-1)} \binom{m}{i}
\end{equation}
\end{theorem}

\section{Finding a label for masks in $M(s,m)$}
In order to use the recursively generated mask set $M(s,m)$, it is essential to find bit positions that uniquely represent, i.e. label, each mask. These positions within a mask are then reserved and used as redundant positions. In general, this problem can be solved by greedy search. For $s=2$, however, there is an algorithm for finding label positions.  

\subsection{Finding a label for masks in $M(2,m)$}

Represent a set of masks in $M(2,m)$ as a matrix with rows being the masks in the set. Every bit position corresponds to a column in the matrix. As in a Reed-Muller code generator matrix, every column corresponds to a point in a binary $m$-dimensional space  $X = \mathbb{F}_2^m = \{x_0,x_1,\ldots, x_{2^m-1} \}$ and 
is represented as a binary $m$-tuple. For example:
\begin{equation*}
M(2,3) = 
\begin{bmatrix}
0 & 0 & 0 & 0 & 0 & 0 & 0 & 0 \\
0 & 0 & 0 & 0 & 1 & 1 & 1 & 1 \\
0 & 0 & 1 & 1 & 0 & 0 & 1 & 1 \\
0 & 1 & 0 & 1 & 0 & 1 & 0 & 1 \\
1 & 1 & 1 & 1 & 1 & 1 & 1 & 1 \\
1 & 1 & 1 & 1 & 0 & 0 & 0 & 0 \\
1 & 1 & 0 & 0 & 1 & 1 & 0 & 0 \\
1 & 0 & 1 & 0 & 1 & 0 & 1 & 0 
\end{bmatrix},
\end{equation*}
with column $0$ corresponding to point $x_0=\left(0,0,0\right)$, column $1$ corresponding to point $x_1=\left(0,0,1\right)$, etc.

Observe that the bottom half of the masks is just a complement of the top half of the masks. Therefore, the label of each of the top masks will be a complement of the label of one of the bottom masks and the two will be different and distinguishable. Also, observe that masks consisting of all zeros and all ones will have a label all zero and all one respectively no matter what label positions are chosen. Hence, we can focus on top half of the masks without all-zero mask. For $M(2,3)$, we have:
\begin{IEEEeqnarray}{rCl}
S(2,3) &=& 
\begin{bmatrix}
0 & 0 & 0 & 0 & 1 & 1 & 1 & 1 \\
0 & 0 & 1 & 1 & 0 & 0 & 1 & 1 \\
0 & 1 & 0 & 1 & 0 & 1 & 0 & 1 \\
\end{bmatrix} \nonumber \\
\text{column number:} &  & \;\;\,  \begin{matrix}
0 & 1 & 2 & 3 & 4 & 5 & 6 & 7 \\
\end{matrix} \nonumber 
\end{IEEEeqnarray}

For the subset of masks in $S(s,m)$, we need to find a set of bits such that each mask has a unique label, but also that none of the labels is a complement of another label. If this happens, the complements of the labels will appear in the bottom half of the masks and will not be unique. Finally, none of the labels should be all-zero or all-one since those are reserved for all-zero and all-one masks. 

Note  that each column in this subset corresponds to binary column number. In this case we have a total of $8$ masks. In general, we have a total of $2(m+1)$ masks. Therefore, we need $1+\lceil \log_2 (m+1) \rceil = 3$ label bits. For example, one such choice could be:
\begin{center}
\begin{tabular}{ |c|c|c| } 
 \hline
 0 & 0 & 1 \\ 
\hline
 0 & 1 & 0 \\ 
\hline
 0 & 1 & 1 \\ 
 \hline
\end{tabular}
\end{center}
The columns correspond to numbers $0, 3$ and $5$ and those are the label positions as shown in gray in the full mask set:
\begin{IEEEeqnarray}{rCl}
\label{masked-set-s2}
M(2,3)=  \left[\begin{array}{>{\columncolor{gray!20}}ccc>{\columncolor{gray!20}}cc>{\columncolor{gray!20}}ccc}
0 & 0 & 0 & 0 & 0 & 0 & 0 & 0 \\
0 & 0 & 0 & 0 & 1 & 1 & 1 & 1 \\
0 & 0 & 1 & 1 & 0 & 0 & 1 & 1 \\
0 & 1 & 0 & 1 & 0 & 1 & 0 & 1 \\
1 & 1 & 1 & 1 & 1 & 1 & 1 & 1 \\
1 & 1 & 1 & 1 & 0 & 0 & 0 & 0 \\
1 & 1 & 0 & 0 & 1 & 1 & 0 & 0 \\
1 & 0 & 1 & 0 & 1 & 0 & 1 & 0 
  \end{array}\right]
\end{IEEEeqnarray}

The general construction of the label set is given in Algorithm \ref{alg:labels2}.

\begin{algorithm} 
\caption{Construction of Label for $M(2,m)$}\label{alg:labels2}
INPUT: Mask set $M(2,m)$ represented as $N(2,m) \times 2^m$ matrix. \\
OUTPUT: Set of label columns $\left(l_0, l_1, \ldots, l_{L-1}\right)$, with $l_j \in \{0,1,\ldots, 2^m-1\}$ for $j=0,1,2,\ldots, L-1$, and $L = 1+  \left\lceil \log_2(m+1) \right\rceil$

\begin{algorithmic}[1]  
\State Reserve the two binary strings of length $L$:
		\[
		0\cdots0 \text{ for all-zero mask}, \quad 1\cdots1 \text{ for all-one mask}.
		\]
\State  Select $m$ distinct non-complementary binary strings of length $L$ from the remaining $2^L-2$ strings. Assign these to the rows \(1, 2, \dots, m\) in $M(2,m)$. Denote the label assigned to row $k$ by \(\ell(k) \in \{0,1\}^L\). For example, we can use the binary representation of numbers, so that \(\ell(k) = (k)_{\text{base } 2}\).
\State For each complement mask, i.e. row \(m+1, m+2, \dots, 2m\), assign
		\[
		\ell(m+k) = \text{bitwise complement of } \ell(k).
		\]
\State For each bit position \(j=0,\dots,L-1\) in a binary string of length $L$, define the column
		\[
		c^{(j)} = \big( (\ell(1))_j, (\ell(2))_j, \dots, (\ell(m))_j \big) \in \{0,1\}^m
		\]
		where \((\ell(k))_j\) is the $j$-th bit of \(\ell(k)\). 
\State For each \(j=0,\dots,L-1\), label column position $l_j$ corresponding to column $c^{(j)}$ is a number such that $(l_j)_{\text{base} 2} = c^{(j)}$.

\end{algorithmic}  
\end{algorithm}
The Algorithm \ref{alg:labels2} guarantees the following:
	\begin{itemize}
		\item By construction, the restriction of row $k$ to columns \(c^{(1)},\dots,c^{(L)}\) equals its assigned label \(\ell(k)\), which is distinct from all other labels.
		\item The restriction of \(k+m\) equals the bitwise complement of \(\ell(k)\), also distinct from all others.
		\item The all-zero and all-one masks are distinguished by the reserved labels \(0\cdots0\) and \(1\cdots1\).
	\end{itemize}
Hence, the restriction map to the selected columns is injective, and every vector receives a unique label.
	
There are many possible labels for mask set $M(2,m)$. From the construction it follows that we can use any $m$ distinct $L$-tuples, $L= 1+ \lceil \log_2 (m+1) \rceil$, as long as they are not all-zero, nor all-one, and as long as no two $L$-tuples are complements of each other. Any permutation of a label set is considered the same label. Therefore, there are $\left( \prod_{i=1}^m (2^L-2i) \right) / L!$ possible labels. Also, the label is guaranteed to exist if the number of choices is at least $m$, i.e. if $(2^L-2)/2 \geq m$. However, this is true since $2^{L-1} \geq m+1$.

\subsection{Finding a label for masks in $M(s,m)$ with $s>2$}

For $s>2$, finding a label is a more challenging problem. A greedy search can provide a label for reasonable code lengths, but does not guarantee minimum label size. Given a set of masks, label search or label construction is done once, and the found set of label positions is used as stuck-at redundancy without a need to be recomputed. Hence, finding or constructing a label does not contribute to any increase in encoding or decoding complexity or latency. 

Since the label construction is not available for $s>2$, the exact label size is not known. However, an upper and lower bound exist and are given in the following two theorems. For full proof, see Appendix~\ref{app:proof-thm-label-min}.

\begin{theorem}
\label{thm-label-min}
Consider the set of masks $M(s,m)$ covering  $s>2$ stuck bits in a binary sequence of length $2^m$ generated by the recursion in (\ref{eq:recursive_union}). Let $L_{\min}(s,m)$ be a label of minimum size that uniquely represents each mask in this set. The following inequalities holds:
\begin{IEEEeqnarray}{rCl}
\label{eq:thm-label-min-I}
L_{\min}(s,m) \geq 2^{s-2} (1 + \log_2 (m-s+3)).
\end{IEEEeqnarray}
\begin{IEEEeqnarray}{rCl}
L_{\min}(s,m) &\leq& 2^{s-1} ( 2 \log_2 N_M(s,m) - 1)  +1 \label{eq:thm-label-min-II} \\
&\leq& 1 + 2^{s-1} + 2^s ( s-1) (1 + \log_2 m).  \label{eq:thm-label-min-III}
\end{IEEEeqnarray}
\end{theorem}

\section{Reed-Muller codes for joint random and stuck-at error correction}

Theorem \ref{thm-masks-in-RMcode} shows that the masks in the set $M(s,m)$ generated by the recursive construction in (\ref{eq:recursive_union}) are codewords in any Reed-Muller code $RM(r,m)$ with $s-1 \leq r$. This directly implies that any codeword in $RM(r,m)$ for $r \geq s-1$ masked with a mask in $M(s,m)$ remains a codeword in the same code $RM(r,m)$ and its full error correcting capability of $RM(r,m)$ can be used. The encoding and decoding are done as described in Fig. \ref{fig-block-diagram-jnt-stuck-rand-error-encoder} and Fig. \ref{fig-block-diagram-jnt-stuck-rand-error-decoder}, and Algorithms \ref{alg:encII} and \ref{alg:decII}.

During encoding, before forming a codeword that needs to be masked, we must ensure that the mask label positions remain equal to zero after encoding, so that after masking those positions indicate what mask was used. One way to make sure label positions remain zero after encoding is to have the label positions correspond to a part of a systematic portion of the $RM$ codeword. The minimum number of linearly dependent columns in the generator matrix of $RM(r,m)$ code is equal to the minimum distance of its dual code $RM(m-r-1,m)$, $2^{r+1}$. Therefore, if a label is found of size $L< 2^{r+1}$, its bit positions are guaranteed to be independent and a part of a systematic portion of the $RM$ code used. However, if $L \geq 2^{r+1}$, we need to make sure that the corresponding columns in the generator matrix of $RM$ code used are indeed linearly independent. Once an intermediate codeword is formed and stuck bit positions and their values are available, appropriate covering mask can be generated in a recursive manner.

On the decoder side, the mask used needs to be recovered based on the label bits extracted from the decoded codeword. Storing the entire mask set ensures quick retrieval of the mask corresponding to a specific label in cases where we have relatively small mask set. However, due to recursive construction, it may be possible to avoid storing the entire mask set.  For example, for a given label of the mask set, groups of $s$ label bits can be used to generate sets of candidate masks and the final mask can be determined as the intersection of all candidate mask sets. This would add some additional complexity to the decoder.

As an encoding example, consider the mask set $M(2,3)$ shown in (\ref{masked-set-s2}) with label bit positions $0$, $3$, and $5$. Since $M(2,3) \subset RM(2,3)$, we can use $RM(2,3)$ code, i.e. a single parity check code with $n=8$ and $k=7$, for single error detection, and bits at positions $0$, $3$ and $5$ to cover $2$ stuck-at errors. This results in an $(8,4)$ code. Upon receiving user data $(u_0,u_1,u_2,u_3)=(1,1,0,1)$, we form an intermediate message by inserting zeros in positions corresponding to mask label $(v_0,v_1,\ldots,v_6)=(\underline{0}, 1, 1, \underline{0}, 0, \underline{0}, 1)$.
The intermediate codeword is obtained by encoding the intermediate message: $\mathbf{b} = (\underline{0}, 1, 1, \underline{0}, 0, \underline{0}, 1, 1)$. Assume that the side information is available and that there are stuck-at defects at bit positions $2$ and $5$, both at value $1$. Therefore, we need to find a mask that has values $0$ and $1$ at positions $2$ and $5$, respectively. One such mask is $\mathbf{m}=(0,0,0,0,1,1,1,1)$. Adding it to the intermediate codeword, we obtain the final codeword  $\mathbf{b} = (\underline{0}, 1, \mathbf{1}, \underline{0}, 1, \underline{\mathbf{1}}, 0, 0)$. Note that the final codeword satisfies stuck-at defect values (bit values in bold), and is also a codeword in $RM(2,3)$ code. The underlined bits represent a mask label and based on them we can uniquely identify which mask in the set $M(2,3)$ was used. On the decoding side, if no errors are detected in a read word, based on the label positions we determine the mask estimate $\hat{\mathbf{m}}$. Adding it to the received word forms an intermediate codeword estimate $\hat{\mathbf{b}}$. The message estimate $\hat{\mathbf{u}}$ is obtained by removing label bits and single parity check bit from the intermediate codeword estimate.

A table of all codes found  for $s \leq 4$ and $m \leq 12$ is given in Table \ref{tab1}. The redundancy, i.e. the size of the found label, is denoted by  $r$, while the lower bound on the redundancy from Theorem~\ref{thm-label-min} is denoted by $r_{lb}$. It is worth pointing out that the labels for $s=3$ and $s=4$ are found by greedy and non-exhaustive search and better, i.e. smaller size labels could exist.

\begin{table}[htbp]
\caption{Reed-Muller Codes for stuck-at and random error correction}
\begin{center}
\begin{tabular}{|c|c|c|c|c|c|}
\hline
$m$ & $s$  & $n$ & $N_M(s,m)$  & $r$ & $r_{lb}$ \\
\hline
\hline
\multirow{2}{*}{3} & 2 & \multirow{2}{*}{8}  & 8 & 3 &  3\\
\cline{2-2} \cline{4-6} 
& 3 & & 34 & 6 & 6 \\
\hline
\multirow{3}{*}{4} & 2 & \multirow{3}{*}{16}  & 10 & 4 & 4 \\
\cline{2-2} \cline{4-6} 
& 3 & & 60 & 8 & 6\\
\cline{2-2} \cline{4-6} 
& 4 & & 246 & 14 & 12\\
\hline
\multirow{3}{*}{5} & 2 & \multirow{3}{*}{32}  & 12 & 4 & 4 \\
\cline{2-2} \cline{4-6} 
& 3 & & 94 & 9 & 8\\
\cline{2-2} \cline{4-6} 
& 4 & & 536 & 19 & 12\\
\hline
\multirow{3}{*}{6} & 2 & \multirow{3}{*}{64}  & 14 & 4 &  4\\
\cline{2-2} \cline{4-6} 
& 3 & & 136 & 9 & 8\\
\cline{2-2} \cline{4-6} 
& 4 & & 996 & 24 & 16\\
\hline
\multirow{3}{*}{7} & 2 & \multirow{3}{*}{128}  & 16 & 4 & 4 \\
\cline{2-2} \cline{4-6} 
& 3 & & 186 & 10 & 8\\
\cline{2-2} \cline{4-6} 
& 4 & & 1666 & 25 & 16\\
\hline
\multirow{3}{*}{8} & 2 & \multirow{3}{*}{256}  & 18 & 5 & 5 \\
\cline{2-2} \cline{4-6} 
& 3 & & 244 & 12 & 8\\
\cline{2-2} \cline{4-6} 
& 4 & & 2586 & 27 & 16\\
\hline
\multirow{3}{*}{9} & 2 & \multirow{3}{*}{512}  & 20 & 5 &  5\\
\cline{2-2} \cline{4-6} 
& 3 & & 310 & 13 & 10\\
\cline{2-2} \cline{4-6} 
& 4 & & 3796 & 31 & 16\\
\hline
\multirow{3}{*}{10} & 2 & \multirow{3}{*}{1024}  & 22 & 5 & 5 \\
\cline{2-2} \cline{4-6} 
& 3 & & 384 & 13 & 10\\
\cline{2-2} \cline{4-6} 
& 4 & & 5336 & 31 & 20\\
\hline
\multirow{3}{*}{11} & 2 & \multirow{3}{*}{2048}  & 24 & 5 &  5\\
\cline{2-2} \cline{4-6} 
& 3 & & 466 & 13 & 10\\
\cline{2-2} \cline{4-6} 
& 4 & & 7246 & 32 & 20\\
\hline
\multirow{3}{*}{12} & 2 & \multirow{3}{*}{4096}  & 26 & 5 & 5 \\
\cline{2-2} \cline{4-6} 
& 3 & & 556 & 14 & 10\\
\cline{2-2} \cline{4-6} 
& 4 & & 9566 & 33 & 20\\
\hline
\end{tabular}
\label{tab1}
\end{center}
\end{table}

Let's look at a few examples from the Table \ref{tab1}  in more detail.

{\it Example 1:}
The set $M(3,6)$ consists of $136$ masks that can cover any $3$ stuck-at positions at any binary values in a sequence of length $n=64$. Fig.~\ref{mask_set_s3_m6} shows the set of masks as rows of black and white pixels where black pixel corresponds to $1$ and white pixel corresponds to $0$. It reveals the fractal structure of the mask set. To distinctly label $136$ masks in the $M(3,6)$ set, we found the following $9$ bit locations using greedy search: $(0, 3, 14, 20, 25, 43, 50, 60, 63)$. In this way we obtain a $(64,55)$ code that can cover any $3$ stuck-at errors among $64$ bits. All the masks in $M(3,6)$ belong to a $RM(2,6)$ code, but also to higher order codes, $RM(3,6)$, $RM(4,6)$ and $RM(5,6)$. This offers extra random error protection for $t=7$ if $RM(2,6)$ is used, $t=3$, if $RM(3,6)$ is used, $t=1$ is $RM(4,6)$ is used and single random error detection if $RM(5,6)$ is used. In each case, total redundancy is combined random and stuck-at redundancy. 

\begin{figure}[htbp]
\centerline{\includegraphics[width=4.5cm]{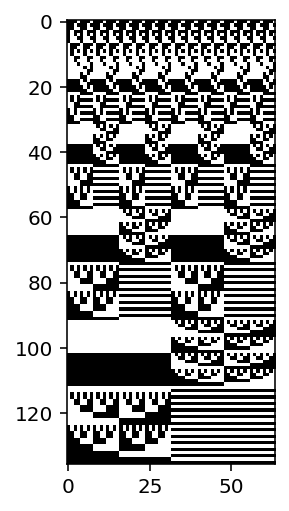}}
\caption{Fractal structure of the recursively constructed mask set, shown in the example of $M(3,6)$}
\label{mask_set_s3_m6}
\end{figure}

{\it Example 2:}
The set $M(3,10)$ consists of $384$ masks that can cover any $3$ stuck-at positions at any binary values in a sequence of length $n=1024$. To distinctly label $384$ masks in the $M(3,10)$ set, we found the following $13$ bit locations using greedy search: $(0, 2, 76, 112, 255, 339, 410, 421, 555, 662, 797, 870, 952)$. In this way we obtain a $(1024,1011)$ code that can cover any $3$ stuck-at errors among $1024$ bits. All the masks in $M(3,10)$ belong to a $RM(r,10)$ code with $r=2,3,4,5,6,7, 8$, offering joint random error protection for $t=127, 63, 31, 15, 7, 3, 1$, respectively. If $RM(9,10)$ is used only a single error detection is possible.

{\it Example 3:}
The set $M(4,9)$ consists of $3796$ masks that can cover any $4$ stuck-at positions at any binary values in a sequence of length $n=512$. To distinctly label $3796$ masks in the $M(4,9)$ set, we found the following $31$ bit locations using greedy search: $(0,  40,  49,  63,  86,  88,  99, 106, 135, 154, 166, 172, 205, 208, 233,$ $241, 246, 267, 284, 294, 306, 320, 345, 357, 383, 405, 425, 439,$ $451, 462, 508)$. In this way we obtain a $(512,481)$ code that can cover any $4$ stuck-at errors among $512$ bits. All the masks in $M(4,9)$ belong to a $RM(r,9)$ code with $r=3,4,5,6,7$, offering joint random error protection for $t=31, 15, 7, 3, 1$, respectively. If $RM(8,9)$ is used only a single error detection is possible.

\section{Conclusion}
A recursive construction of a set of masks is developed such that it can satisfy any $s$ stuck-at errors in a $2^m$ binary sequence, when $s \leq m$. We prove that the masks generated in this way are codewords in a Reed-Muller $RM(s-1, m)$ code. The constructed set contains no more than $2^s m^{s-1}$ masks. The stuck-at redundancy is embedded in the mask and uniquely represents each mask in the set. We prove the lower and the upper bound on the stuck-at redundancy in general, and label generation algorithm for $s=2$. The stuck-at code described is a non-linear code, but nevertheless a subcode of an $RM(r,m)$ code with $ r \geq s-1$. The supercode $RM(r,m)$ can be used for additional random error correction with its full error correcting capability. The encoding requires no mask search and is straightforward based on the description of the recursive construction. The decoding is done in a single decoding attempt and requires almost no additional complexity or latency. 

Several open problems remain. A detailed description of implementations of the encoder and decoder are not fully worked out, but due to high structure of the mask set, efficient implementations seem quite possible. A general algorithm for finding label bit positions for any $s$ would also greatly contribute to the quality of this work. Finally, a $RM(s-1,m)$ code has a dual code with minimum distance of $2^s$ and therefore its code space can mask any $2^s-1$ stuck-at errors. This comes, of course, at a much higher redundancy cost, utilizing full $RM(s-1,m)$ code redundancy. However, there may be a way to increase both the currently constructed mask set and stuck-at correction capability, finding the compromise between the two methods, linear (based on the dual code) and the non-linear recursive one described in this work. 

\section{Appendix}

%\appendix
\subsection{Proof of Theorem \ref{thm-max-num-masks} \label{app:proof-thm-max-num-masks}}
\begin{proof}
Let $N_M(s,m)$ denote the number of masks in $M(s,m)$. From the recursive construction of $M(s,m)$ we have the following inequality for the number of elements in $M(s,m)$:
\begin{IEEEeqnarray}{rCl}
N_M(s,m) &\leq& N_M(s,m-1) + \nonumber \\
&\sum_{i=1}^{s-1}& N_M(i,m-1) N_M(s-i,m-1) \nonumber \\
&=& N_M(s,m-2) +   \nonumber \\
&\sum_{j=1}^2& \sum_{i=1}^{s-1} N_M(i,m-j) N_M(s-i,m-j) \nonumber \\
\ldots &=& N_M(s,l) +   \nonumber \\
& \sum_{j=1}^{m-l}& \sum_{i=1}^{s-1} N_M(i,m-j) N_M(s-i,m-j). \nonumber
\end{IEEEeqnarray}
We can prove the bound on $N_M(s,m)$ by induction. Assume that $N_M(t,m) \leq 2^t m^{t-1}$ for all $t<s$ and all $m \geq \lceil \log_2 t \rceil$. Observe that this is true for $t=1$  as $N_M(1,m) = 2 \leq 2$ and  $t=2$ as $N_M(2,m)=2(m+1) \leq 2^2 m$ for all $m \geq 1$. Therefore, for $s > 2$ we have:
\begin{IEEEeqnarray}{rCl}
N_M(s,m) &\leq& 2^{2^l}+ \sum_{j=1}^{m-l} \sum_{i=1}^{s-1} 2^s (m-j)^{s-2} \nonumber \\
&=& 2^{2^l}+  2^s (s-1) \sum_{j=1}^{m-l} (m-j)^{s-2} \nonumber \\
&=& 2^{2^l}+ 2^s (s-1) \sum_{j=l}^{m-1} j^{s-2} \nonumber \\
&\leq& 2^{2^l}+ 2^s (s-1) \int_l^{m} x^{s-2} dx \nonumber \\
&=&  2^{2^l}+ 2^s \left( m^{s-1} -l^{s-1} \right) \nonumber \\
& = & 2^{2^l}-2^s l^{s-1} + 2^s m^{s-1}  \nonumber \\
& \leq & 2^s m^{s-1}. \nonumber 
\end{IEEEeqnarray}
When $s \geq 2$, we have the following inequalities: $l \leq \log_2 s+ 1$, $l \geq 2$ and $l \geq 2^{s-1}$, and consequently $2^{2^l} \leq 2^s l^{s-1}$. This concludes the proof.
\end{proof}

\subsection{Proof of Theorem \ref{thm-masks-subsets} \label{app:proof-thm-masks-subsets}}
\begin{proof}
It is easy to see that the theorem holds for $s=2$ and all $m$. Assume $M(t-1,m) \subset M(t,m)$ for all $2 \leq t < s$ and all $m$.  
Denote with $M(t,m) | M(t,m)$ a set of masks of length $2^{m+1}$ obtained by simply concatenating each mask in $\mathbf{m} \in M(t,m)$ with itself, i.e. as $\left[ \mathbf{m}, \mathbf{m} \right]$.
Also, denote with $M(t,m)  \bigotimes M(j,m)$ a set of masks of length $2^{m+1}$ obtained by concatenating each mask in $M(t,m)$ with each mask in $M(j,m)$. 
With this notation, we can represent recursive construction of $M(s,m)$ in (\ref{eq:recursive_union}) as shown in the right hand side of Fig.~\ref{fig5-thm2}.
\begin{figure}[htbp]
\centerline{\includegraphics[width=5.5cm]{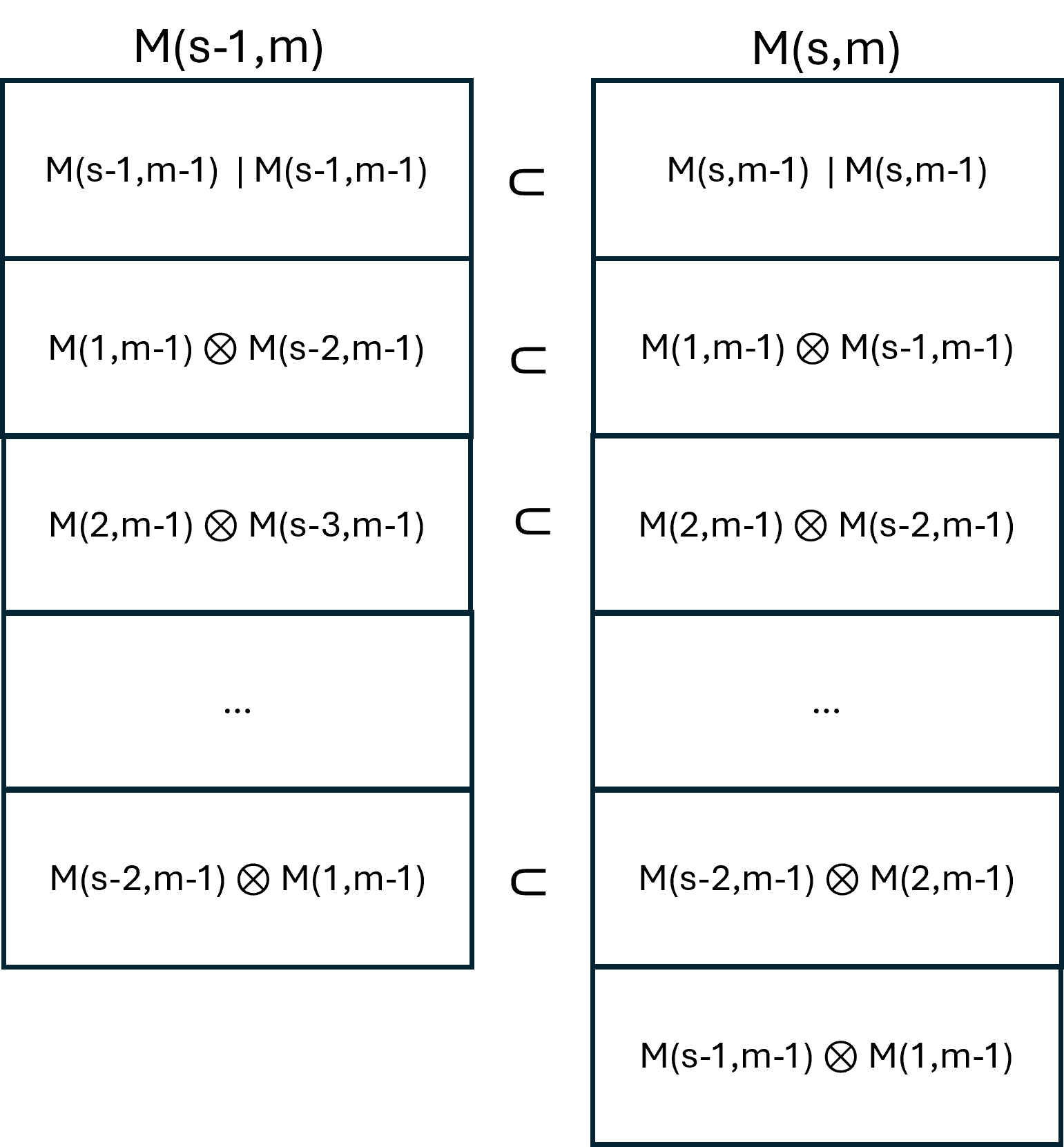}}
\caption{Relationship between a set of masks  of length $2^m$ that can cover any $s$ and any $s-1$ stuck errors}
\label{fig5-thm2}
\end{figure}
Since $M(s-j-1,m-1) \subset M(s-j,m-1)$ for $1 \leq j \leq s-2$, all subsets of masks in $M(s-1,m)$ of type $M(j,m-1) \bigotimes M(s-1-j, m-1)$ are subsets of $M(s,m)$. 
It remains to show that $M(s-1,m-1) \subset M(s,m-1)$. Following the recursion, we again convince ourselves that all subsets of masks in $M(s-1,m-1)$ of type $M(j,m-2) \bigotimes M(s-1-j, m-2)$ are subsets of $M(s,m-1)$. 
Next, we need to show that $M(s-1, m-2) \subset M(s, m-2)$. We continue with the recurision until we either reach mask set of length $2^l$ with $l=\lceil \log_2 s \rceil$, and since $s>2$, $l \geq 2$. By construction, $M(s,l)$ is a set of all $2^l$-tuples, and therefore it must contain all masks in $M(s-1,l)$. This concludes the proof. 
\end{proof}

\subsection{Proof of Theorem \ref{thm-masks-in-RMcode} \label{app:proof-thm-masks-in-RMcode}}
\begin{proof}
We proceed with induction. As we have seen, the theorem holds for $s=1$ and $s=2$ for any $m$. Assume the theorem holds for all sets of masks covering $t$ stuck bits $t<s$, where $s>2$. We need to show that the theorem is also valid for the set of masks covering $s$ stuck bits. Consider any mask of type $\left[ \mathbf{m}_1, \mathbf{m}_2 \right]$ with $\mathbf{m_1} \in M(i,m-1)$, $\mathbf{m}_2 \in M(s-i, m-1)$, and $i=1,2,\ldots,s-1$ in the construction in (\ref{eq:recursive_union}). Masks $\mathbf{m}_1$ and $\mathbf{m}_2$ are codewords in $RM(i-1,m-1)$ and $RM(s-i-1,m-1)$, respectively, and therefore also codewords in $RM(s-2,m)$. From the recursive construction of Reed-Muller code, we conclude that  $\left[ \mathbf{m}_1, \mathbf{m}_2 \right]$ is a codeword in $RM(s-1, m)$. We need to show that any mask of the type $\left[ \mathbf{m}, \mathbf{m} \right]$ with $\mathbf{m} \in M(s,m-1)$ in also in $RM(s-1,m)$. To show this, it is sufficient to show that $\mathbf{m}$ is in $RM(s-1,m-1)$, because by connecting a codeword in $RM(s-1,m-1)$ with itself, we get a codeword in $RM(s-1,m)$. Consider a mask $\mathbf{m} \in M(s,m-1)$. Using recursive construction of a mask set, we know that it either belongs to the set of masks $\left[ \mathbf{m}_1, \mathbf{m}_2 \right]$ with $\mathbf{m_1} \in M(i,m-2)$, $\mathbf{m}_2 \in M(s-i, m-2)$, and $i=1,2,\ldots,s-1$ or is of a type $\left[ \mathbf{m}, \mathbf{m} \right]$ with $\mathbf{m} \in M(s,m-2)$. In the former case, we conclude that the masks are codewords in $RM(s-1,m-1)$. In the latter case, we have to show that $\left[ \mathbf{m}, \mathbf{m} \right]$ with $\mathbf{m} \in M(s,m-2)$ is a codeword in $RM(s-1,m-1)$, which is done by showing that $\mathbf{m} \in M(s,m-2)$ is a codeword in $RM(s-1,m-2)$. We continue using the recursion until we reach the set of masks in $M(s,l)$ with $l=\lceil \log_2 s \rceil$. From the construction, we know that mask set $M(s,l)$ consists of all possible $2^l$-tuples, i.e. it constitutes $RM(s,s)$ code. This concludes the proof.
\end{proof}

\subsection{Proof of Theorem \ref{thm-masks-in-G-RMcode} \label{app:proof-thm-masks-in-G-RMcode}}
\begin{proof}
The set $M(1,m)$ contains rows in the generator matrix of $RM(0,m)$, i.e. a single all-one vector of length $2^m$. The set $M(1,m)$ also contains all-zero vector, i.e. the complement of the generator matrix. Furthermore, masks in $M(2,m)$ also contain all rows of $G(1,m)$ and their complements. Due to recursive construction of the masks and Plotkin's resursive construction of the generator matrix of $RM$ code, we conclude that the set of masks covering $s$ stuck bits in a binary sequence of length $2^m$ generated by the recursion in (\ref{eq:recursive_union}) contains all the rows in the Plotkin's generator matrix of Reed-Muller code of order $s-1$ and length $2^m$, $RM(s-1,m)$, and their complements. 
\end{proof}

\subsection{Proof of Theorem \ref{thm-label-min} \label{app:proof-thm-label-min}}
\begin{proof}[Proof of (\ref{eq:thm-label-min-I})]
From the construction in \ref{eq:recursive_union}, a set of masks $M(s,m)$ contains a mask of type $\left[ \mathbf{m}, \mathbf{m}^{\prime} \right]$ where $\mathbf{m}$ can be any mask in $M(s-1,m-1)$. It also contains a mask of type $\left[ \mathbf{m}^{\prime}, \mathbf{m} \right]$ where $\mathbf{m}$ can be any mask in $M(s-1,m-1)$. Consequently, any label set of size $L$, $(l_0, l_1, \ldots, l_{L-1})$, must be able to uniquely define a mask in $\mathbf{m} \in M(s-1,m-1)$ both on the left and on the right hand side. Therefore:
\begin{IEEEeqnarray}{rCl}
L(s,m) & \geq & 2 L(s-1,m-1) \geq 2^2 L(s-2, m-2) \nonumber \\
& \cdots  & \geq 2^{s-2} L(2,m-s+2) \nonumber \\
&= & 2^{s-2} (1+ \lceil \log_2 (m-s+3) \rceil). \nonumber 
\end{IEEEeqnarray}
\end{proof}

\begin{proof}[Proof of (\ref{eq:thm-label-min-II}) and (\ref{eq:thm-label-min-III})]
To prove the second part of the theorem, we recognize that any valid label must distinguish any pair of masks. To simplify the notation, let $N=N_M(s,m)$, and $n=2^m$. The number of different pairs of masks is $N_p=N(N-1)/2$. Create a $N_p \times n$ grid, with each row corresponding to a pair of masks and each column corresponding to a mask bit position. Check each location in the grid where coresponding mask bit position distinguishes betwen coresponding pair of masks. We start selecting mask bit positions to be placed in a set of label bits. Every time we select a position, we eliminate all the rows that are checked by that bit position, i.e. that have a check in the corresponding column. Next, we observe that $M(s,m)$ is a subset of $RM(s-1, m)$ with minimum distance $d=2^{m-s+1}$. This implies that every row in the grid has at least $d=2^{m-s+1}$ check marks. Denote with $\gamma_i$ an average checks per column in the grid before selection of the $i$-th label bit, and with $\bar{d}_i$ an average checks per row in a grid before selection of the $i$-th label bit. We have:
\[ \gamma_0 n = N_p \bar{d}_0 \]
\[ \gamma_0 = N_p \bar{d}_0/n \]
Hence, we can find a column with at least $\gamma_0$ checks, choose it to be our first label column, and remove at least $\gamma_0$ rows. This leaves us with at most $N_p (1-\bar{d}_0/n)$ rows. In the next step we have:
\[ \gamma_1 = N_p (1-\bar{d}_0/n) \bar{d}_1/n \]
Therefore, there is a column with at least $\gamma_1$ checks, choose it to be our second label column, and remove at least $\gamma_1$ rows. This leaves us with at most $N_p (1-\bar{d}_0/n)(1-\bar{d}_1/n)$ rows. After $L$ steps and $L$ label bit selections, we will have at most $N_p \prod_{i=0}^{L-1} (1-\bar{d}_i/n)$ rows. Since each average checks per row is at least $d$, the largest number of remaining rows is less than $N_p (1-d/n)^L$. The label will distinguish all pairs if the number of ramaining rows is less than 1, i.e. if:
\begin{IEEEeqnarray}{rCl}
& & N_p (1-d/n)^L < 1 \nonumber \\
& \Leftrightarrow & (N^2/2) (1-d/n)^L < 1 \nonumber \\
& \Leftrightarrow & 2 \log_2 N + L \log_2 (1-d/n) -1 < 0 \nonumber \\
& \Leftrightarrow & 2 \log_2 N - L d/n-1 < 0 \nonumber \\
& \Leftrightarrow & L > 2^{s-1} (2 \log_2 N_M(s,m) -1)
\end{IEEEeqnarray}
We conclude that there exists a label of size $L \leq  2^{s-1} ( 2 \log_2 N_M(s,m) - 1)+1$ that distinguishes all pairs of masks and therefore uniquely represents each mask. The second inequality folows directly from Theorem \ref{thm-max-num-masks}.
\end{proof}

\bibliographystyle{IEEEtran}
\bibliography{RM4jnt_stuck_rnd_err_coding.bib}

%\end{thebibliography}

\end{document}